\documentclass{llncs}
\usepackage{amssymb,amsmath}
\usepackage{graphicx, epsfig}
\usepackage{paralist}
\usepackage{bbm}

\def\floor#1{\lfloor #1 \rfloor}
\def\ceil#1{\lceil #1 \rceil}

\newcommand{\no}[1]{}

\newcommand{\settwo}[2]{\ensuremath{\{\,#1\,|\,#2\,\} }}

\newtheorem{fact}{Fact}
\renewenvironment{proof}{\trivlist\item[]\emph{Proof}:}%
{\unskip\nobreak\hskip 1em plus 1fil\nobreak$\Box$
\parfillskip=0pt%
\endtrivlist}

\newenvironment{itemize*}%
  {\begin{itemize}%
    \setlength{\itemsep}{0pt}%
    \setlength{\parskip}{0pt}%
    \setlength{\parsep}{0pt}%
    \setlength{\topsep}{0pt}%
    \setlength{\partopsep}{0pt}%
  }%
  {\end{itemize}}%

\newcommand{\cT}{{\cal T}}

\newcommand{\cL}{{\cal L}}
\newcommand{\cK}{{\cal K}}

\newcommand{\oR}{{\overline R}}
\newcommand{\oV}{{\overline V}}
\newcommand{\oP}{{\overline P}}

\newcommand{\eps}{\varepsilon}
\newcommand{\tpred}{\mathrm{tpred}}
\newcommand{\bT}{\ensuremath{\mathbbm{T}}}
\newcommand{\Col}{\mathrm{Col}}
\newcommand{\hmina}{h_{\min}}
\newcommand{\hmaxa}{h_{\max}}
\newcommand{\ohmin}{\overline{h_{\min}}}
\newcommand{\ohmax}{\overline{h_{\max}}}

\newcommand{\hmin}{\mathit{min}}
\newcommand{\hmax}{\mathit{max}}
\newcommand{\gmin}{\mathit{gmin}}
\newcommand{\gmax}{\mathit{gmax}}

\begin{document}
\title{Optimal Color Range Reporting in One Dimension}
\author{ 
   Yakov Nekrich\inst{1}  and Jeffrey Scott Vitter\inst{1} }
\institute{
   The University of Kansas. {\tt yakov.nekrich@googlemail.com, jsv@ku.edu}
}


\date{}
\maketitle

\begin{abstract}
Color (or categorical) range reporting is a variant of the orthogonal range reporting problem in which every point in the input is assigned a \emph{color}. While the answer to an orthogonal point reporting query contains all points in the query range $Q$, the answer to a color reporting query contains only  distinct colors of points in $Q$. 
In this paper we describe an $O(N)$-space data structure that answers one-dimensional color reporting queries in optimal $O(k+1)$ time, where $k$ is the number of colors in the answer and $N$ is the number of points in the data structure. Our result can be also dynamized and extended to the external memory model.
\end{abstract}

\section{Introduction}
 In the orthogonal range 
reporting problem, we store a set of points $S$ in a data structure so that 
for an arbitrary range $Q=[a_1,b_1]\times \ldots\times [a_d,b_d]$ all 
points from $S\cap Q$ can be reported. Due to its importance, one- and multi-dimensional range reporting was extensively studied in computational geometry and database communities. The following situation frequently arises in different areas of computer science: a set of $d$-dimensional objects $\{\,( t_1,t_2,\ldots,t_d)\,\}$ must be preprocessed so that we can enumerate all objects satisfying $a_i\le t_i \le b_i$ for arbitrary $a_i,b_i$, $i=1,\ldots,d$. This scenario can be modeled by the orthogonal range reporting problem. 

The objects in the input set can be  distributed into
\emph{categories}. Instead of enumerating all objects,  we may want  to report  distinct categories of objects in the given range.  This situation can be modeled by the color (or categorical) 
range reporting problem: every point in a set $S$ is assigned a color (category); we pre-process $S$, so that for any $Q=[a_1,b_1]\times \ldots\times [a_d,b_d]$ the distinct colors of points in $S\cap Q$ can be reported.

Color range reporting  is usually considered to be a more complex problem than point reporting. 
For one thing, we do not want to report the same color multiple times. In this paper 
we show that complexity gap can be closed for one-dimensional color range reporting. We describe color reporting data structures with the same space usage and query time as the best known corresponding structures for point reporting. \no{Furthermore, our  static one-dimensional  data structure 
is optimal.} Moreover we extend our result to  the external memory model.

\paragraph{Previous Work.}
We can easily report points in a one-dimensional range 
$Q=[a,b]$ by searching for the successor of $a$ in $S$, 
$succ(a,S)=\min\settwo{e\in S}{e\ge a}$. If $a'=succ(a,S)$
is known, we can traverse the sorted list of points in $S$ starting at $a'$ and report all elements in $S\cap [a,b]$. 
We can find the successor of $a$ in $S$ in $O(\sqrt{\log N/\log \log N})$ time~\cite{BF02}; if the universe size is $U$, i.e., if 
all points are positive integers that do not exceed $U$, then the successor can be 
found in $O(\log \log U)$ time~\cite{BoasKZ77}. 
Thus we can report all points in $S\cap [a,b]$ in $O(\tpred(N)+k)$ time for 
$\tpred(N)=\min(\sqrt{\log N/\log \log N}, \log \log U)$.
Henceforth $k$ denotes the number of elements (points or colors) in the query answer.  
It is not possible to find the successor in $o(\tpred(N))$ time unless the universe size $U$ is very small or the space usage of the data structure is very high; see e.g.,~\cite{BF02}.
However, reporting points in a one-dimensional range  takes 
less time than searching for a successor. 
In their fundamental paper~\cite{MNSW98}, Miltersen et al. showed that one-dimensional point reporting queries can be answered in $O(k)$ time using an $O(N\log U)$ space data structure. 
Alstrup et al. ~\cite{ABR01} obtained  another surprising result:  they presented an $O(N)$-space data structure that answers point reporting queries in $O(k)$ time and thus achieved both optimal query time and optimal space usage for this problem. 
The data structure for one-dimensional point reporting can be 
dynamized so that queries are supported in $O(k)$ time and 
updates are supported in $O(\log^{\eps} U)$ time~\cite{MPP05}; henceforth $\eps$ denotes an arbitrarily small positive constant. We refer to~\cite{MPP05} for further update-query time trade-offs. Solutions of the one-dimensional 
point reporting problem are based on finding an arbitrary element $e$ in a 
query range $[a,b]$; once such $e$ is found, we can traverse the sorted list of points until all points in $[a,b]$ are reported. Therefore it is straightforward to extend  point reporting results  to the external memory model.

Janardan and Lopez~\cite{JL93} and Gupta et al.~\cite{GuptaJS95} showed that one-dimensional color reporting queries can be answered in $O(\log N + k)$ time, both in the static and the dynamic scenarios. Muthukrishnan~\cite{Muth02} 
described a static $O(N)$ space  data structure that answers queries in $O(k)$ time if all point coordinates are bounded by $N$. 
We can obtain  data structures that use $O(N)$ space and answer queries in $O(\log\log U + k)$ or $O(\sqrt{\log N/\log \log N}+k)$ time using the reduction-to-rank-space technique. 
No data structure that answers one-dimensional color reporting 
queries in $o(tpred(N))+ O(k)$ time was previously known.  
A dynamic data structure of Mortensen~\cite{M03} supports queries and 
updates in $O(\log \log N + k)$ and $O(\log \log N)$ time respectively if the values of all elements are bounded by $N$.

Recently, the one- and two-dimensional color range reporting problems in the external memory model were studied in several papers~\cite{LarsenP12,Nekrich12,LarsenvW13}. Larsen and Pagh~\cite{LarsenP12} described a data structure that uses linear space and answers one-dimensional color reporting queries in 
$O(k/B +1)$ I/Os if values of all elements are bounded by $O(N)$. In the case when values of elements are unbounded the best previously known data structure needs $O(\log_B N + k/B)$ I/Os to answer a query; this result can be obtained by combining the data structure from~\cite{ArgeSV99} and  reduction of one-dimensional color reporting to three-sided\footnote{A three-sided range query is a two-dimensional orthogonal range query that is open on one side. For instance, queries $[a,b]\times [0,c]$ and $[a,b]\times [c,+\infty]$ are three-sdied queries.} point reporting~\cite{GuptaJS95}. 

In another recent paper~\cite{ChanDSW12}, Chan et al.  described a data structure that supports the following queries on a set of points whose values are bounded by $O(N)$:  
for any query point $q$ and any integer~$k$, we can report the first $k$  colors 
that occur after $q$. This data structure can be combined with the result 
from~\cite{ABR01} to answer queries in $O(k+1)$ time. Unfortunately, the solution 
in~\cite{ChanDSW12} is based on the hive graph data structure~\cite{Chazelle86}. 
Therefore it cannot be used to solve the problem in external memory or to obtain a dynamic solution. 

\no{
The paper of Mortensen~\cite{M03} describes a reduction of  one-dimensional color reporting problem to the problem of reporting two-dimensional uncolored points in a three-sided range$Q=[a,b]\times [0,c]$ such that the $y$-coordinates of uncolored points are bounded by $\log N$. This reduction can be used to obtain a data structure that uses $O(N)$}

\paragraph{Our  Results.}
As can be seen from the above discussion and Table~\ref{tab:res}, there are significant complexity gaps  between color reporting and point 
reporting data structures in one dimension. We show in this paper that it is possible to close these gaps. 

In this paper we show that one-dimensional color reporting queries can be answered in constant time per reported color for an arbitrarily large size of the  universe.
Our data structure uses $O(N)$ 
space and supports color reporting queries in $O(k+1)$ time. 
This data structure can be dynamized so that query time and space usage remain unchanged; the updates are supported in $O(\log^{\eps} U)$ time where $U$ is the size of the universe. 
\no{Our two-dimensional data structure uses $O(N\log^{\eps}N)$ words 
of space and answers queries in $O(\log \log U +k)$ time. 
We can further reduce the space usage; two other  data structures use $O(N)$ and $O(N\log \log N)$ space and answer 
two-dimensional queries in $O((k+1)\log^{\eps}N)$ and $O(\log \log U + k\log\log N)$ time.} The  new results are listed at the bottom of  Table~\ref{tab:res}.

Our internal memory results  are valid in the word RAM model of computation, the same model that was used in e.g.~\cite{ABR01,MPP05,Muth02}. In this model, we assume that any standard arithmetic operation and the basic bit operations can be performed in constant time. We also assume that each word of memory consists of $w\ge \log U\ge \log N$ bits, where $U$ is the size of the universe. That is, we make a reasonable and realistic assumption that the value of any element fits into one word of memory. 

Furthermore, we also extend our data structures to the external memory model. Our 
static data structure uses linear space and answers color reporting queries in $O(1+k/B)$ I/Os. 
Our dynamic external data structure also has optimal space usage and query cost; updates are supported in $O(\log^{\eps}U)$ I/Os.   
\no{Our dynamic data structure produces a list of colors that may contain several occurrences of the same color; however, this can happen only if the main memory size $M$ is very small. }

\begin{table*}[htb]
  \centering
  \begin{tabular}{|l|l|c|c|c|c|}
    \hline
Ref.\ &    Query & Space & Query & Universe  & Update\\ 
      &    Type  & Usage & Cost  &       & Cost    \\ \hline
\cite{ABR01}  &   Point Reporting  & $O( N)$  & $O(k+1)$ &  & static  \\
\cite{MPP05}  &   Point Reporting  & $O(N)$  & $O(k+1)$ &  & $O(\log^{\eps}U)$ \\
\hline
\cite{GuptaJS95,JL93}  &    Color Reporting  & $O( N)$  & $O(\log N +k)$ &  & $O(\log N)$ \\
\cite{Muth02}  &    Color Reporting  & $O( N)$  & $O(k+1)$ & $N$ &  static \\
\cite{Muth02}  &    Color Reporting  & $O( N)$  & $O(\log \log U + k)$ & $U$ & static  \\ 
\cite{Muth02}  &    Color Reporting  & $O( N)$  & $O(\sqrt{\log N/\log \log N} + k)$ &  & static  \\  
\cite{M03}  &    Color Reporting  & $O( N)$  & $O(\log \log N + k)$ & $N$ & $O(\log \log N)$ \\
\hline  
Our & Color Reporting  & $O(N)$  & $O(k+1)$ &  & static  \\
Our & Color Reporting  & $O(N)$  & $O(k+1)$ &  & $O(\log^{\eps}U)$\\ \hline
\end{tabular}
\caption{Selected previous results and new results for one-dimensional color reporting. The fifth  and the sixth row can be obtained by applying the reduction to rank space 
to the result from~\cite{Muth02}. 
}
  \label{tab:res}
\end{table*}

\section{Static Color Reporting in One Dimension} \label{sec:statcol}
We start by describing a  static data structure that uses $O(N)$ space 
and answers color reporting queries in $O(k+1)$ time. 

All elements of a set $S$ are stored in a balanced binary tree $\cT$. Every leaf of $\cT$, except for the last one,  contains $\log N$ elements, the last leaf contains at most $\log N$ elements,  and every internal node has two children.
For any node $u\in \cT$, $S(u)$ denotes the set of all elements stored in the leaf descendants of $u$. 
For every color $z$ that occurs in $S(u)$, the set $Min(u)$ ($Max(u)$) contains the minimal (maximal) element $e\in S(u)$ of color $z$. The list $L(u)$ contains 
the $\log N$ smallest elements of $Min(u)$ in increasing order. The list 
$R(u)$ contains the $\log N$ largest elements of $Max(u)$ in decreasing order.
For every internal non-root node $u$  we store the list $L(u)$ if $u$ is the right child of its parent; if $u$ is the left child of its parent, we store the list $R(u)$ for $u$. All lists $L(u)$ and $R(u)$, $u\in \cT$, contain $O(N)$ elements in total since the tree has $O(N/\log N)$ internal nodes.

We define the middle value $m(u)$ for an internal node $u$ as the minimal 
value stored in the right child of $u$, $m(u)=\min\{\, e\,|\, e\in S(u_r)\,\}$
where $u_r$ is the right child of $u$. The following \emph{highest range ancestor} query plays a crucial role in the data structures of this and the following sections. 
The answer to the highest range ancestor query 
$(v_l,a,b)$ for a leaf $v_l$ and values $a<b$ is the highest ancestor $u$ of $v_l$, such that $a< m(u)\le b$; if $S\cap [a,b]=\emptyset$, the answer is undefined. 
The following fact elucidates the meaning of the highest range ancestor.
\begin{fact}\label{fact:eluc}
Let $v_a$ be the leaf that holds the smallest $e\in S$, such that $e\ge a$;
let $v_b$ be the leaf that holds the largest $e\in S$, such that $e\le b$.
\no{Let $w$ be the lowest common ancestor of $v_a$ and $v_b$.} 
Suppose that  $S(v_l)\cap [a,b]\not=\emptyset$ for some leaf $v_l$ and  $u$ is the answer to the highest range ancestor query $(v_l,a,b)$. Then  $u$ 
is the lowest common ancestor of $v_a$ and $v_b$. 
\end{fact}
\begin{proof}
  Let $w$ denote the lowest common ancestor of $v_a$ and $v_b$. 
  Then $v_a$ and $v_b$ are in $w$'s left and right subtrees respectively. Hence, $a< m(w) \le b$ and $w$ is not an ancestor of $u$. If $w$ is a descendant of $u$ and $w$ is in the right subtree of $u$, then $m(u)\le a$. If $w$ is in the left subtree of $u$, then 
$m(w)>b$.
\end{proof} 
We will show that we can find $u$ without searching for $v_a$ and $v_b$ and answer highest range ancestor queries on a balanced tree 
in constant time.

For every leaf $v_l$, we store two auxiliary data structures.  All elements of $S(v_l)$ are 
stored in a data structure $D(v_l)$ that uses $O(|S(v_l)|)$ space and answers 
color reporting queries on $S(v_l)$ in $O(k+1)$ time. 
We also store a data structure $F(v_l)$ that 
uses $O(\log N)$ space; for any $a<b$, such that $S(v_l)\cap [a,b]\not=\emptyset$, $F(v_l)$  answers the 
highest range ancestor query $(v_l,a,b)$ in $O(1)$ time.   Data structures $D(v_l)$ and $F(v_l)$ will be described later in this section.   Moreover, we store all elements of $S$ 
in the data structure described in~\cite{ABR01} that supports one-reporting queries:  for  any $a<b$, some  element 
$e\in S\cap [a,b]$ can be found in  $O(1)$ time; 
if $S\cap [a,b]=\emptyset$, the data structure returns a dummy element $\perp$. 
Finally, all elements of $S$ are stored in a slow data structure that uses 
$O(N)$ space and answers color reporting queries in $O(\log n+ k)$ time. We can use e.g. the data structure from~\cite{JL93} for this purpose. 

\paragraph{Answering Queries.}
All colors in a query range $[a,b]$ can be reported with the following procedure. 
Using the one-reporting data structure from~\cite{ABR01},  we search for  some  $e\in S\cap[a,b]$ if at least one such $e$ exists. 
If no element $e$ satisfying $a\le e\le b$ is found, then $S\cap [a,b]=\emptyset$ and  the query is answered. 
Otherwise, let $v_e$ denote the leaf that contains $e$.
\no{ 
If all elements $e$, $a\leq e\leq b$, are contained in $S(v_e)$ and 
$S(v_n)$ where $v_n$ is the leaf that follows $v_e$ (respectively, in 
$S(v_e)$ and $S(v_p)$ where $v_p$ is the leaf that precedes $v_e$), then 
we can answer the query using the structures for $S(v_e)$ and $S(v_n)$ 
(respectively the structures for $S(v_e)$ and $S(v_p)$). Otherwise 
all elements of $S(v_e)$ belong to the interval $[a,b]$. }
Using $F(v_e)$, we search for the highest ancestor~$u$ of $v_e$ such that $a\le m(u)\le b$. If no such $u$ is found, then all $e$, 
$a\le e\le b$, are in $S(v_e)$. We can report all colors in 
$S(v_e)\cap [a,b]$ using $D(v_e)$. If $F(v_e)$ returned some node 
$u$, we proceed as follows.
Let $u_l$ and $u_r$ denote the left and the right children of $u$. 
We traverse the list $L(u_r)$ until an element $e'> b$ is found 
or the end of $L(u_r)$ is reached. We also traverse $R(u_l)$  until an element $e'< a$ is found or the end of $R(u_l)$ is reached.
If we reach neither the end of $L(u_r)$ nor the end of $R(u_l)$, then 
the color of every encountered element $e\in L(u_r)$, $e\le b$, and 
$e\in R(u_l)$, $e\ge a$, is reported. Otherwise the range 
$[a,b]$ contains at least $\log N$ different colors. In the latter  case we can use any data structure for one-dimensional color range reporting~\cite{JL93,GuptaJS95}
 to identify all colors from $S\cap[a,b]$ in $O(\log n+k)=O(k+1)$ time.
\begin{figure*}[tb]
  \centering
  \includegraphics[width=.6\textwidth]{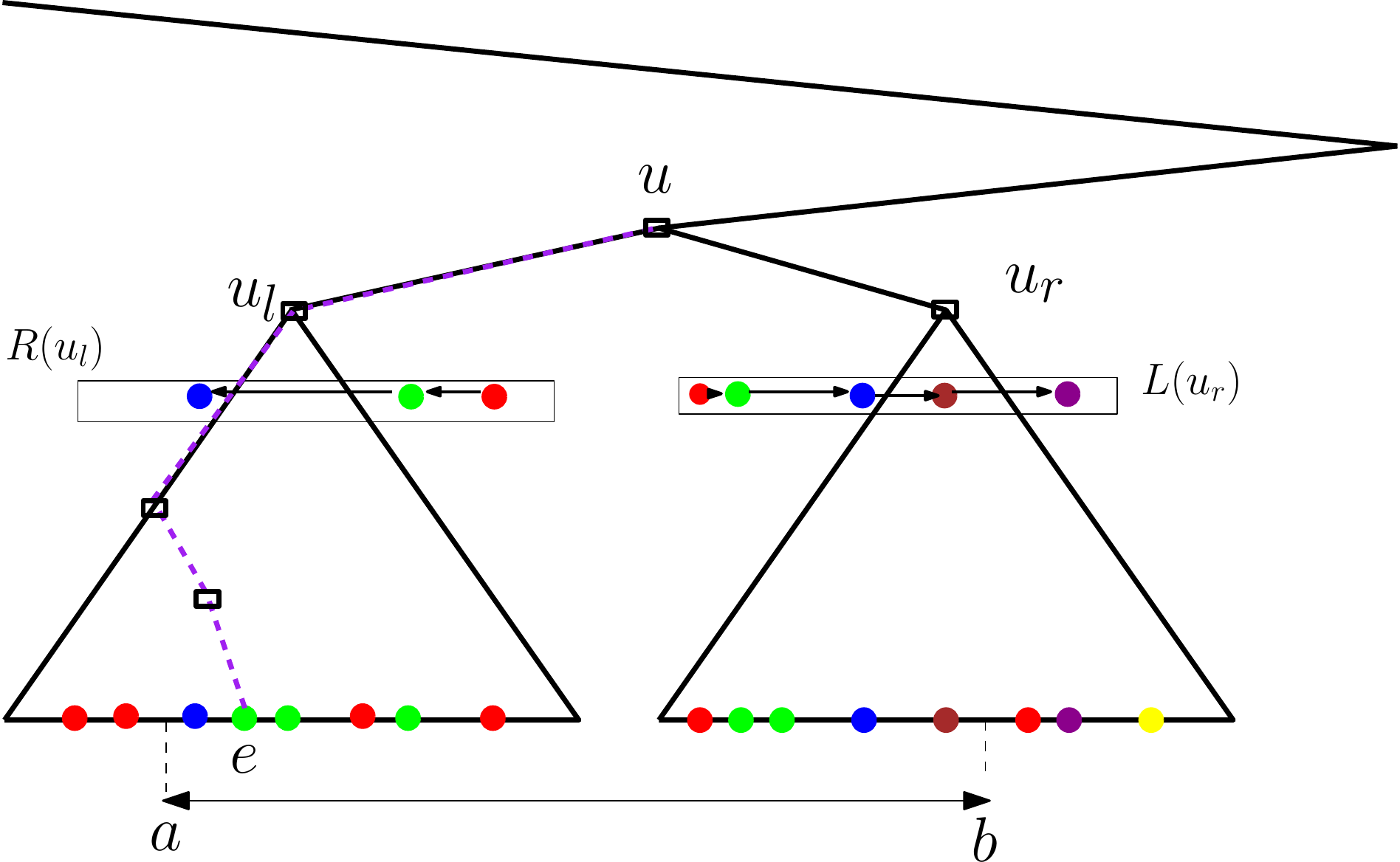}
  \caption{Answering a color reporting query $Q=[a,b]$: $e$ is an arbitrary element in $S\cap [a,b]$, $u$ is the highest range ancestor of the leaf that contains $e$, the path from $e$ to $u$ is indicated by a dashed line. We assume $\log n=5$, therefore $L(u_r)$ contains $5$ elements and the yellow point is not included in $L(u_r)$. To simplify the picture, we assumed that each leaf contains only one point\no{do not show the leaves of $\cT$}; only relevant parts of $\cT$ are on the picture. }
  \label{fig:prtree}
\end{figure*}

\paragraph{Leaf  Data Structures.} 
A data structure $D(v_l)$ answers color reporting queries on $S(v_l)$ as follows.  In~\cite{GuptaJS95}, the authors show how 
a one-dimensional color reporting query on a set of $m$   one-dimensional elements can be answered by answering a query $[a,b]\times [0,a]$ on a set of $m$ uncolored two-dimensional points. 
A standard priority search tree~\cite{McCreight85} enables us to answer queries of the form $[a,b]\times [0,a]$ on $m$ points 
in $O(\log m)$ time. 
Using a combination of fusion trees and priority search trees, 
described by Willard~\cite{W00}, we can answer queries 
 in $O(\log m/\log \log N)$ time. The data structure of Willard~\cite{W00} uses $O(m)$ space and a universal look-up table of size $O(\log^{\eps}N)$ for an arbitrarily small $\eps$. Updates 
are also supported in $O(\log m/\log \log N)$  time\footnote{In~\cite{W00}, Willard only considered queries on $N$ points, but extension to the case of any $m\le N$ is straightforward.}.

Since $S(v_l)$ contains $m=O(\log N)$ elements, we can answer colored queries on $S(v_l)$ in $O(\log m/\log \log N)=O(1)$ time.  Updates 
are also supported in $O(1)$ time; this fact will be used in  Section~\ref{sec:dyncol}.

Now we describe how $F(v_l)$ is implemented. Suppose that $S(v_l)\cap [a,b]\not=\emptyset$ for some leaf $v_l$. Let $\pi$ be the path from $v_l$ to the root of $\cT$. We say that a  node $u\in\pi$ is a \emph{left parent} if $u_l\in\pi$ for the left child $u_l$ of $u$; a  node $u\in\pi$ is a \emph{right parent}
if $u_r\in\pi$ for the right child $u_r$ of $u$. 
If $S(v_l)$ contains at least one $e\in [a,b]$, then the following is true.
\begin{fact}\label{fact:val}
If $u\in\pi$ is a left parent, then $m(u)> a$. If $u\in\pi$ is a right parent, then $m(u)\le b$.
\end{fact}
\begin{proof} 
If $u\in \pi$ is a left parent, then $S(v_l)$ is in its left subtree. 
Hence, $m(u)$ is greater than  any $e\in S(v_l)$ and $m(u)> a$. 
If $u$ is the right parent, than $S(v_l)$ is in its right subtree. Hence, $m(u)$ is smaller than or equal to any $e\in S(v_l)$ and $m(u)\le  b$.
\end{proof}
\begin{fact}\label{fact:mono}
If $u_1\in \pi $ is a  left parent and $u_1$ 
is an ancestor of $u_2\in\pi$, then $m(u_1)> m(u_2)$. 
If $u_1\in \pi $ is  a right parent and $u_1$ 
is an ancestor of $u_2\in\pi$, then $m(u_2)> m(u_1)$. 
\end{fact}
\begin{proof}
If $u_1$ is a left parent, then $u_2$ is in its left subtree. 
Hence, $m(u_1) > m(u_2)$ by definition of $m(u)$. 
If $u_2$ is a right parent, then $u_1$ is in its right subtree. 
Hence, $m(u_1) < m(u_2)$ by definition of $m(u)$.
\end{proof}
Suppose that we want to find the highest range ancestor of $v_l$ for a range $[a,b]$ such that $S(v_l)\cap [a,b]\not=\emptyset$.
Let $\cK_1(\pi)$ be the set of  middle values $m(u)$ for left parents $u\in \pi$ sorted by height;
let $\cK_2(\pi)$ be the set of $m(u)$ for right parents $u\in \pi$ sorted by height. 
By Fact~\ref{fact:mono}, elements of $\cK_1$ ($\cK_2$) increase (decrease) monotonously. 
By Fact~\ref{fact:val}, $m(u)> a$ for any $m(u)\in \cK_1$ and 
$m(u)<b$ for any $m(u)\in \cK_2$.
Using fusion trees~\cite{FW94}, we can search in 
$\cK_1$ and find the highest node $u_1\in\pi$ such that $u_1$ is a
left parent and $m(u_1)\le b$. We can also search in $\cK_2$
and find the highest node $u_2\in\pi$ such that $u_2$ is a
right parent and $m(u_2)> a$. Let $u$ denote the higher node 
among $u_1$, $u_2$. Then $u$ is the  highest ancestor of $v_l$ 
such that $m(u)\in [a+1,b]$.

\tolerance=1000
\paragraph{Removing Duplicates.}
When a query is answered, our  procedure returns a color $z$ two times if $z$  occurs in both $S\cap [a,m(u)-1]$ and $S\cap [m(u),b]$. We can easily produce a list without sorting  in which each color occurs exactly once. Let 
$\Col$ denote an array with one entry for every color that occurs 
in a data structure. Initially $\Col[i]=0$ for all $i$. We traverse the list of colors $\cL$ produced by the above described procedure. Every time when we encounter 
a color $z$ in $\cL$ such that $\Col[z]=0$, we set $\Col[z]=1$;
when we  encounter a color $z$ such that $\Col[z]=1$, we remove the corresponding entry from $\cL$. When the query is answered, we traverse 
$\cL$ once again and set $Col[z]=0$ for all $z \in \cL$. 
\begin{theorem}
   \label{theor:col1d}
There exists an $O(N)$-space data structure that supports one-dimensional  color range reporting queries in $O(k+1)$ time.
\end{theorem}

\section{Color Reporting in External Memory}
The static data structure of Section~\ref{sec:statcol} can be used for answering queries in external memory. 
We  only need to increase the sizes of $S(v_l)$, $R(u)$, and $L(u)$ to $B\log_B N$, and use an external memory variant of the slow data structure for color reporting~\cite{ArgeSV99}.  This approach enables us to achieve $O(1+k/B)$ query cost, but one important issue should be addressed. 
\no{it has one important drawback.} 
As explained in Section~\ref{sec:statcol}, the same color can be reported twice when a query is answered. However, we cannot get rid of duplicates in $O(1+k/B)$ I/Os using the method of Section~\ref{sec:statcol} because of its random access to the list of reported colors. Therefore we need to make further changes in our internal memory solution. 
For an element $e\in S$, let $prev(e)$ denote the largest element $e'\le e$  of the same color. For every element $e$ in  $L(u)$ and any  $u\in T$, we also store the value of $prev(e)$.

\no{
 We modify our data structure as follows. 

Each set $S(v_l)$ for a leaf  $v_l$ contains $B^2\log_B^2N$ points.  The list 
$R(u)$ contains $B\log_B N$ rightmost points from $Max(u)$. Instead of a list 
$L(u)$, we store the following array of lists. Let $Rcol_i(u)$ denote the set of colors of the first $i$ points in $R(u)$. The list $L_i(u)$ contains the first 
$B\log_B N$ points $p$  from $Min(u)$ such that $col(p)\not\in Rcol_i(u')$ for the left sibling $u'$ of $u$. All lists $R(u)$ and $L_i(u)$ contain $O(N)$ points. 
}

We define each set $S(v_l)$ for a leaf  $v_l$ to contain $B\log_BN$ points. Lists $L(v)$ and $R(v)$ for an internal node $v$ contain $B\log_B N$ leftmost points from $Min(v)$ (respectively, 
$B\log_B N$ rightmost points from $Max(v)$).
Data structures $F(v_l)$ are implemented as in Section~\ref{sec:statcol}. A data structure $D(v_l)$ supports color reporting queries on $S(v_l)$ and is implemented as follows. We can answer a one-dimensional color reporting query by answering a three-sided point reporting query on a set $\Delta$ of $|S(v_l)|$ two-dimensional points; see e.g.,~\cite{GuptaJS95}.
If $B\ge\log_2 N$, $S(v_l)$  and $\Delta$ contain $O(B^2)$ points. In this case we can use the  data structure from~\cite{ArgeSV99} that uses linear space and answers three-sided queries in $O(\log_B |S(v_l)|+k/B)=O(1+k/B)$ I/Os.
 If $B<\log N$, $S(v_l)$ and $\Delta$ contain $O(\log^2 N)$ points.  Using the data structure from~\cite{FW94}, we can find the predecessor of any value $v$ in a set of $O(\log^2 N)$ points  in $O(1)$ I/Os. Therefore we can apply the rank-space technique~\cite{GabowBT84} and reduce  three-sided 
point reporting queries on $\Delta$ to three-sided point reporting queries on a grid of size $|\Delta|$ (i.e., to the case when coordinates of all points are integers bounded by $|\Delta|$) using a constant number of additional I/Os. Larsen and Pagh~\cite{LarsenP12} described a linear space data structure  that answers three-sided point reporting queries for $m$ points  on an $m\times m$ grid  in $O(1+k/B)$  I/Os.  Summing up, we 
can answer a three-sided query on a set of $B\log_B N$ points in $O(1+k/B)$ I/Os. Hence, we can also answer 
a color reporting query on $S(v_l)$ in $O(1+k/B)$ I/Os using linear space.

A query $Q=[a,b]$ is answered as follows. We find the highest range ancestor $u$ for any $e\in S\cap [a,b]$ exactly as in Section~\ref{sec:statcol}. If $u$ is a leaf, we answer the query using $D(u)$. Otherwise the reporting procedure proceeds as follows. 
We traverse the list $R(u_l)$ for the left child $u_l$ of $u$ until some point $p< a$ is found. If $e\ge a$ for all $e\in R(u_l)$, then there are at least $B\log_B N$ different colors in $[a,b]$ and we can use a slow data structure to answer a query in $O(\log_B N + \frac{k}{B})=O(1+\frac{k}{B})$ I/Os. 
Otherwise we traverse $L(u_r)$ and report all elements $e$ such that 
$prev(e)< a$. If $prev(e)\ge a$ for $e\in L(u)$, then an element of the same color was reported when $R(u_l)$ was traversed. Traversal of $L(u_r)$ stops when 
an element $e>b$ is encountered or the end of $L(u_r)$ is reached. In the former case, we reported all colors in $[a,b]$. In the latter case the number of colors in $[a,b]$ is at least $B\log_B N$. This is because every element in $L(u_r)$ 
corresponds to a distinct color that occurs at least once in $[a,b]$. Hence, we 
can use the slow data structure and answer the query in $O(\log_B N +\frac{k}{B})=O(\frac{k}{B}+1)$ I/Os. 
\begin{theorem}
   \label{theor:col1dext}
There exists a linear-space data structure that supports one-dimensional  color range reporting queries in $O(k/B+1)$ I/Os.
\end{theorem}

\no{ 
Otherwise suppose that $i$ rightmost points in $R(u_l)$ are greater than or equal to $a$. We traverse the list $L_i(u_r)$ until an element $e>b$ is encountered or the end of $L_i(u_r)$ is reached. If we reach the end of  $L_i(u_r)$, then the total number of colors in $[a,b]$ is at least $B\log_B N$ and we can use a slow data structure to answer queries as described above. Otherwise all distinct colors that occur in $[a,b]$ are the colors of points in the traversed parts of $R(u_l)$ and $L_i(u_r)$.
}

\section{Base Tree for Dynamic Data Structure}
\label{sec:dyncol}
In this section we show how the base tree and auxiliary data structures of the static solution can be 
modified for usage in the dynamic scenario. 
To dynamize the data structure of Section~\ref{sec:statcol},
we slightly change the balanced tree $\cT$ and secondary data structures: every leaf of $\cT$ now contains $\Theta(\log^{2}N)$ elements of $S$ and 
each internal node has $\Theta(1)$ children.  
We store the lists $L(u)$ and $R(u)$ in each internal non-root node of $u$. We associate several values $m_i(u)$ to each node $u$: for every child~$u_i$ of~$u$, except  the leftmost 
child $u_1$, $m_i(u)=\min\{\,e\,|\,e\in S(u_i)\,\}$. 
The highest range ancestor of a leaf $v_l$  is the highest 
ancestor $u$ of $v_l$ such that $a < m_i(u) \le b$ for at 
least one $i\not= 1$. Data structures $D(v_l)$ and $F(v_l)$ 
are defined as in Section~\ref{sec:statcol}. 
We also maintain a data structure of~\cite{MPP05} that reports an 
arbitrary element $e\in S\cap [a,b]$ if the range $[a,b]$ is not empty.

We implement the base tree $\cT$ as the weight-balanced B-tree~\cite{AV03} 
with the leaf parameter $\log N$ and the branching parameter $8$. This means that every internal node has between $2$ and $32$ children and each leaf contains between $2\log^2 N$ and $\log^2 N$ elements. Each internal non-root node on level $\ell$ of $\cT$ has between $2\cdot 8^{\ell}\log^2 N$ and $(1/2)\cdot 8^{\ell}\log^2 N$ elements in its subtree. If the number of elements in some node $u$ exceeds $2\cdot 8^{\ell}\log^2 N$, we 
split $u$ into two new nodes, $u'$ and $u''$. In this case 
we insert a new  value $m_i(w)$ for the parent $w$ of $u$. 
Hence, we may have to update the data structures $F(v_l)$ 
for all leaf descendants of $w$. A weight-balanced B-tree is engineered in such a way that a split occurs at most once in a sequence of $\Omega(8^{\ell}\log^2 N)$ insertions (for our choice 
of parameters). Since $F(v_l)$ can be updated in $O(1)$ time, 
the total amortized cost incurred by splitting nodes is $O(1)$. When an element $e$ is deleted, we delete it from the set $S(v_l)$. 
\no{When an element $e$ is deleted and}
If  $e=m_i(u)$ for a deleted element $e$ and some node 
$u$, we do not change the value of $m_i(u)$. We also do not 
start re-balancing if some node contains too few elements in its subtree. But we re-build the entire tree $\cT$ if the total number of deleted elements equals  $n_0/2$, where 
$n_0$ is the number of elements that were stored in $\cT$ when 
it was built  the last time. Updates can be de-amortized without increasing the cost of update operations by scheduling the procedure of re-building nodes (respectively, re-building the tree)~\cite{AV03}
\paragraph{Auxiliary Data Structures.}
We implement $D(v_l)$ in the same way as in Section~\ref{sec:statcol}. Hence color
queries on $S(v_l)$ are answered in $O(\log |S(v_l)|/\log \log N)=O(1)$ time and updates are also supported in $O(1)$ time~\cite{W00}. 

 We need to modify data structures $F(v_l)$, however, because  $\cT$ is not a binary tree in the dynamic case. Let $\pi$ denote a path from $v_l$ to the root
for some leaf $v_l$. 
We say that a node $u$ is an $i$-node if $u_i\in \pi$ for the 
$i$-th child $u_i$ of $u$. 
\begin{fact}\label{fact:val2}
Suppose that $S(v_l)\cap [a,b]\not=\emptyset$ and $\pi$ is the path from $v_l$ to the root. If $u\in \pi$ is an $i$-node, then $m_j(u)< b$ for $1\le j\le i$
and $m_j(u)>a$ for $j>i$.
\end{fact}
We say that a value $m_j(u)$ for  $u\in \pi$ is a \emph{left value}
if $j\le i$ and $u$ is an $i$-node.  A value  $m_j(u)$ for  $u\in \pi$ is a \emph{right value}
if $j> i$ and $u$ is an $i$-node.
\begin{fact}\label{fact:mono2}
If $m_j(u_1)$ is a left value and $u_1\in\pi$ is an ancestor of $u_2\in\pi$, then $m_j(u_1)\le m_f(u_2)$ for any $f$. 
If $m_j(u_1)$ is a right value and $u_1\in\pi$ is an ancestor of $u_2\in\pi$, then $m_j(u_1)> m_f(u_2)$ for any $f$. 
\end{fact}
It is easy to check Facts~\ref{fact:val2} and~\ref{fact:mono2} using the same arguments as in Section~\ref{sec:statcol}.

We store all left values $m_j(u)$, $u\in \pi$, in a set $\cK_1$;
$m_j(u)$ in $\cK_1$ are sorted by the height of~$u$. 
We store all right values $m_j(u)$, $u\in \pi$, in a set $\cK_2$;
$m_j(u)$ in $\cK_2$  are also sorted by the height of $u$. 
Using fusion trees on $\cK_1$, we can find the highest node $u_1$, such that at least one left value $m_g(u_1)> a$. 
We can also find the highest $u_2$  such that 
at least one right value $m_f(u_2)\le b$. 
Since $\cK_1$ and $\cK_2$ contain $O(\log N)$ elements, we can
support searching and updates in $O(1)$ time; see~\cite{FW94,Thorup99}.
By Fact~\ref{fact:val2}, $a< m_g(u_1)\le b$ and $a< m_f(u_2)\le b$. If $u$ is the higher node among $u_1$, $u_2$, then $u$ is an answer to the highest range ancestor query $[a,b]$ for a node $v_l$.

\section{Fast Queries, Slow Updates}
\label{sec:del}
In this section we describe a dynamic  data structure with optimal  query time. Our improvement combines an idea from~\cite{M03} with the highest range ancestor approach. We also use a new solution for a special case of two-dimensional point reporting problem presented in Section~\ref{sec:fast3sid}. 

Let $height(u)$ denote the height of a node $u$.
For an element $e\in S$ let $\hmina(e)=height(u')$, where $u'$ is the highest ancestor of the leaf containing $e$, such that $e\in Min(u')$. We define $\hmaxa(e)$ in the same way with respect to $Max(u)$. All colors in a range  $[a,b]$ can be reported  as follows. We identify 
an arbitrary $e\in S\cap [a,b]$. Using the highest range ancestor data structure, we can find the lowest common ancestor $u$ of the leaves that contain $a$ and $b$. Let $u_f$ and $u_g$ be the children of $u$ that contain the successor of $a$ and the predecessor of $b$. 
Let $a_f=a$, $b_g=b$; let $a_i=m_i(u)$ for  $f< i \le g$ and $b_i=m_{i+1}(u)-1$ for $f\le i < g$.
We can identify unique colors of relevant  points stored in each node $u_j$, $f<j\le g$, by finding all $e\in [a_j,b_j]$ such that $e\in Min(u_j)$.
This condition is equivalent to reporting all $e\in [a_j,b_j]$ such that $\hmina(e)\ge height(u_j)$. We can identify all colors of relevant points in $u_f$ by reporting 
all $e\in [a_f,b_f]$ such that $\hmaxa(e)\ge height(u_f)$. 
Queries of the form $e\in [a,b]$, $\hmina(e)\ge c$, (respectively $e\in [a,b]$, $\hmaxa(e)\ge c$) can be supported using Lemma~\ref{lemma:fast}.
While the same color can be reported several times, we can get rid of duplicates as explained in Section~\ref{sec:statcol}.  

When a new point is inserted into $S$ or when a point is deleted from $S$, we 
can update the values of $\hmina(e)$ and $\hmaxa(e)$ in $O(\log\log U)$ time. 
We refer to~\cite{M03,ShiJ05} for details. \no{It remains to show how} 

While updates of  data structures of Lemma~\ref{lemma:fast} are fast, re-balancing the base tree can be a problem. As described in Section~\ref{sec:dyncol},   when the number of points in a node $u$ on level $\ell$ exceeds  $2\cdot 8^{\ell}\log N$, we split it into two nodes, $u'$ and $u''$.  As a result, the values $\hmina(e)$ for $e$ stored 
in the leaves of $u''$ can be incremented.  Hence, we would have to examine 
the leaf descendants of $u''$ and recompute their values for some of them. Since the 
height of $\cT$ is logarithmic, the total cost incurred by re-computing the 
values $\hmina(e)$ and $\hmaxa(e)$ is $O(\log N)$.
The problem of reducing the cost of re-building the tree will be solved in the following sections.
In Appendix~\ref{sec:slow} we describe another  data structure that 
supports fast updates but answering queries takes polynomial time in the worst case. 
In Section~\ref{sec:fullydyn} we show how the cost of splitting can be reduced by modifying the definition 
of $\hmina(e)$, $\hmaxa(e)$ and using the slow data structure from Appendix~\ref{sec:slow} when the number of reported colors is sufficiently large.

\section{Fast Queries, Fast Updates}
\label{sec:fullydyn}
Let $n(u)$ denote the number of leaves in the subtree of a node $u$.
Let $Left(u)$ denote the set of $(n(u))^{1/2}$ smallest elements in 
$Min(u)$; let $Right(u)$ denote the set of $(n(u))^{1/2}$ largest elements in $Max(u)$.   
We  maintain the values $\ohmin(e)$ and $\ohmax(e)$ for 
$e\in S$, such that for any $u\in \cT$ we have: $\ohmin(e)=\hmina(e)$ if $e\in Left(u)$ and $\ohmin(e)\le \hmina(e)$ if $e\in S(u)\setminus Left(u)$; 
$\ohmax(e)=\hmaxa(e)$ if $e\in Right(u)$ and $\ohmax(e)\le \hmaxa(e)$ 
if $e\in S(u)\setminus Right(u)$.  
We keep $\ohmin(e)$ and $\ohmax(e)$ in data structures of Lemma~\ref{lemma:fast}. We maintain the data structure described in Section~\ref{sec:slow}. This data structure is used to answer queries when the number of colors in the query range is large. It is  also used to update the values of $\ohmin(e)$ and $\ohmax(e)$ when a node is split.  

To answer a query $[a,b]$, we proceed in the same way as in Section~\ref{sec:del}. 
Let $u$, $u_f$, $u_g$, and $a_i$, $b_i$, $f\le i \le g$ be defined as in Section~\ref{sec:del}. 
\no{For a query range $[a,b]$, we find the lowest common ancestor $u$ 
of the leaves $l_a$ and $l_b$ that 
contain the successor of $a$ and the predecessor of $b$ respectively. 
Let $u_f$ and $u_g$ be the children of $u$ that contain $l_a$ and $l_b$ respectively. We report all $e\in [a,b]$, such that $\ohmin(e)\ge height(u_j)$ 
and all $e\in [a,b]$ such that $\ohmax(e)\ge height(u_j)$.
}
Distinct colors in each $[a_i,b_i]$, $f\le i \le g$,  can be reported using the data structure of Lemma~\ref{lemma:fast}. If the answer to at least one of the queries contains at least $(n(u_i))^{1/2}$ elements, then there are 
at least $(n(u_i))^{1/2}$ different colors in $[a,b]$. 
The total number of elements in $[a,b]\cap S$ does not exceed $n(u)=16n(u_j)$. 
Hence, we can employ the data structure from Section~\ref{sec:slow} to 
report all colors from $[a,b]$ in $O(([a,b]\cap S)^{1/2}+k)=O(k)$ time. 
If answers to all queries contain less than $(n(u_i))^{1/2}$ elements, then for every distinct color that occurs in 
$[a,b]$ there is an element $e$ such that $e\in Left(u_i)\cap [a_i,b_i]$, $f\le i <g$, or $[e\in Right(u_g)\cap [a_g,b_g]$. 
By definition of $\ohmin$ and $\ohmax$ we can correctly 
report up to $(n(u_i))^{1/2}$ leftmost colors in $Left(u_i)$ or up to $(n(u_i))^{1/2}$ rightmost colors in $Right(u_i)$.   

\tolerance=500
When a new element $e$ is inserted, we compute the values of 
$\hmina(e)$, $\hmaxa(e)$ and update the values of $\hmina(e_n)$, $\hmaxa(e_n)$, where $e_n$ is the element of the same color as $e$ that follows $e$. This can be done in the same way as in Section~\ref{sec:del}.  When a node $u$ on level $\ell$ is split 
into $u'$ and $u''$, we update the values of $\hmina(e)$ and
$\hmaxa(e)$ for $e\in S(u')\cup S(u'')$. If $\ell \le \log \log N$, we examine all $e\in S(u')\cup S(u'')$ and re-compute the values of $prev(e)$, $\hmina(e)$, and $\hmaxa(e)$. 
Amortized cost of re-building nodes $u$ on $\log\log N$ lowest tree levels is $O(\log \log N)$. 
If $\ell> \log\log N$, $S(u)$ contains $\Omega(\log^5N)$ elements.  We can find $(n(u'))^{1/2}$ elements in 
$Left(u')$, $Left(u'')$, $Right(u')$, and  $Right(u'')$ using 
the data structure from Lemma~\ref{lemma:slowcount}. This takes 
$O((n(u)^{1/2})\log N +\log N\log \log N)=O(((n(u))^{7/10})$ time. 
Since we split a node $u$ one time after $\Theta(n(u))$ insertions, the amortized cost of splitting nodes on level 
$\ell>\log \log N$ is $O(1)$. 
Thus the total cost incurred by splitting nodes after  insertions is $O(\log \log N)$.
Deletions are processed in a symmetric way. 
Thus we obtain the following result
\begin{theorem}
\label{theor:dyn}
  There exists a linear space data structure that supports 
one-dimensional color range reporting queries in $O(k+1)$ time and updates in $O(\log^{\eps}U)$ amortized time. 
\end{theorem}


\setcounter{section}{0}
\renewcommand\thesection{A.\arabic{section}}

\section{Dynamic Three-Sided Queries on a Narrow Stripe}\label{sec:fast3sid}
In this section we describe a data structure that efficiently supports (uncolored) two-dimensional queries in the case when the $y$-coordinate $p.y$ of each point $p\in S$ does not exceed  $\log N$ and the query range is bounded on three sides, two sides in the $x$-dimension and one side in the $y$-dimension; the $x$-coordinate $p.x$ of each $p\in S$  does not exceed the value of parameter $U$. This result is  used in the
dynamic solution of the color range reporting problem. We also believe it to be of independent interest.
\begin{lemma}\label{lemma:fast}
Let $S$ be a set of two-dimensional points such that $p.y\le \log N$ for all $p\in S$. There exists a linear-space data structure that reports all points 
$p\in [a,b]\times [c,\log N]$ in $O(k+1)$ time and supports updates in $O(\log^{\eps}U)$ time, where $U$ is the size of the universe. 
\end{lemma}
\begin{proof}
We divide the points into consecutive groups $G_i$ according to their $x$-coordinates. Each group, except  the last one,  contains $\Theta(\log N)$ points; the last group contains $O(\log N)$ points. For every group $G_i$ and for each value $h\le \log N$, we keep two values $\gmin(i,h)$ and $\gmax(i,h)$. 
Intuitively, we can access the value of $\hmax(i,h)=\max\{\,p.x\,|\, p\in G_i,\, p.y\le h\,\}$ and $\hmin(i,h)=\min\{\,p.x\,|\, p\in G_i,\, p.y\le h\,\}$ for any $h$ and $i$ by examining $O(1)$ values $\gmin(i,h)$ or $\gmax(i,h)$.
We remark that we cannot directly store the values of $\max\{\,p.x\,|\, p\in G_i,\, p.y\le h\,\}$ in $\gmax(i,h)$ because an update operation would be too costly. 

Let $\tau=\ceil{\log^{\eps}N}$.
When a new point $p$ is inserted, we represent $h=p.y$ in base $\tau$ as a sum of powers\footnote{To avoid tedious details, we assume that $\log N$ is a power of $\tau$.}  of 
$\tau$: $h=a_0+a_1\tau+\ldots+a_g\tau^g$ where $0\le a_i<\tau$ and $g\le \floor{1/\eps}$.  
Let  $h_{r,s}$ for $r=g,g-1,\ldots,0$ and $s=1,2,\ldots, a_r$  be defined as $h_{r,s}= \sum_{j=r+1}^g a_j\tau^j+s\cdot \tau^r$. Let $G_i$ denote the group into which $p$ must be inserted. We examine the values $\gmin(i,h_{r,s})$  and $\gmax(i,h_{r,s})$ for all $r=g,g-1,\ldots,0$ and $s=1,2,\ldots, a_r$. If $\gmin(i,h_{r,s})>p.x$ or $\gmax(i,h_{r,s})< p.x$, we update their values.  The deletion procedure is symmetric. In both cases, we examine and update $O(\log^{\eps}N)$ values.

\begin{fact}\label{fact:simp}
Suppose that some $G_j$ contains at least one point $p$ such that $p.x\in [a,b]$ and 
$p.y\ge h$. Suppose further that $\max\{\,p.x\,|\, p\in G_i\,\}\le b$ or $\min\{\,p.x\,|\, p\in G_i\,\}\ge a$. 
Then either $\hmin(i,h)\in [a,b]$ or $\hmax(i,h)\in [a,b]$. 
\end{fact}

\begin{fact}\label{fact:sel}
For any $h$, we can select $O(1)$ values $f_0,f_1,\ldots, f_g$, such that 
for any $i$  there is at least one $j$, $1\le j\le g$, satisfying $\hmin(i,h)=\gmin(i,f_j)$ and at least one $l$, $1\le l\le g$, satisfying $\hmax(i,h)=\gmax(i,f_l)$.
\end{fact}
\begin{proof}
Let $h= a_g\tau^g+\ldots+a_1\tau+a_0$  where $0\le a_i<\tau$ . 
We set $f_0=h$, $f_1=\sum_{s=2}^g a_s\tau^s+(a_1+1)\tau$, 
$\ldots$, $f_v=\sum_{s=v+1}^{g} a_s\tau^s +(a_v+1)\tau^v$, $\ldots$, 
$f_g=(a_g+1)\tau^g$.
If $\hmin(i,h)=a$, then there is a point $p$, such that $p.y=h'\ge h$ and 
$p.x=a$. Let $h'= a'_g\tau^g+\ldots+a'_1\tau+a_0$. If $h'\not=h$, we consider an index $t$ such that 
$a'_t>a_t$ and $a_r=a'_r$ for $r<t$.  
By definition of $\gmin(\cdot,\cdot)$, \no{we set the  value $\gmin(i,h'_{t,a_t+1})=p.x$ when 
$p$ was inserted. Here}
$\gmin(i,h'_{t,a_t+1})=p.x$, where  $h'_{t,a_t+1}=\sum_{s=t+1}^{g} a_s\tau^s +(a_t+1)\tau^t$. 
By our choice of $f_i$,  $h'_{t,a_t+1}=f_t$. Hence, $\gmin(i,f_t)=\hmin(i,h)$. 
If $h'=h$, then $\gmin(i,h)=p.x$ and $\gmin(i,f_0)=\hmin(i,h)$.
The statement concerning $\hmax(i,h)$ can be proved in the same way.
\end{proof}

Facts~\ref{fact:simp} and~\ref{fact:sel} suggest the following method for answering queries.  
For any $h$, $1\le h\le \log N$, we store all $\gmin(i,h)$ and $\gmax(i,h)$ in a data 
structure $R_h$ that supports one-dimensional point reporting queries.  Data structure 
$\oR$ contains $x$-coordinates of all points in $S$ and also supports one-dimensional reporting queries. Using the result from~\cite{MPP05},  $\oR$ and all $R_h$ support queries 
in $O(k+1)$ time and updates in $O(\log^{\eps}N)$ time.  For each $G_i$, we maintain a data structure $H_i$. For any $a\le b$ and $1\le c\le \log N$, $H_i$ can report all points $p\in G_i$, $a\le p.x\le b$ and $c\le p.y\le \log N$. $H_i$ uses $O(|G_i|)$ space and supports updates in $O(1)$ time. We can implement $G_i$ in the same way as the data structure 
$D(v_l)$ in Section~\ref{sec:statcol}.  

Given a query $Q=[a,b]\times [c,\log N]$, we use $\oR$ for identifying an arbitrary $p_0\in S$, 
such that $p_0.x\in [a,b]$. Let $G_0$ be the group that contains $p_0$. If all points 
$p$ such that $p.x\in [a,b]$ are contained in $G_0$, we use the data structure $H_0$ to answer 
the query. If $[a,b]$ spans several groups, then we generate the values $f_0,\ldots,f_g$ as in the proof of Fact~\ref{fact:sel}. We query data structures $R_{f_j}$ and identify all 
values $\gmin(i,f_j)\in [a,b]$ and $\gmax(i,f_j)\in [a,b]$.  For every such $\gmin(i,f_j)$ and $\gmax(i,f_j)$, we visit the corresponding group $G_i$ and answer the query $[a,b]\times [c,\log N]$ on $G_i$ using $H_i$. 

By Facts~\ref{fact:simp} and~\ref{fact:sel},  queries to $R_{f_j}$ return at least  one representative element from every $G_i$ such that $G_i\cap Q\not=\emptyset$ and only such groups will be visited. 
Since we ask $g+1$ queries to $R_{f_j}$, every group $G_i$ is visted at most $g+1=O(1)$ times. 

When a new point $p$  is inserted into $S$, we identify the group $G_i$ where it belongs  in $O(\log\log N)$ time. Then we re-examine the values of $\gmin(i,h_{r,s})$  and $\gmax(i,h_{r,s})$ for $h=p.y$. If necessary, we update the data structures 
$R_{h_{r,s}}$. 
When the number of elements in some $G_i$ becomes equal to $2\log N$, we split $G_i$ into two groups in a standard way.  Deletions are symmetric. 
Thus updates are supported in $O((\log^{\eps}N)\log^{\eps}U)$ time and queries can be answered in 
$O(1)$ time. We obtain the result of Lemma~\ref{lemma:fast} if we replace $\eps$ by $\eps/2$ into the above proof.
\end{proof}

\no{
We can extend the result of Lemma~\ref{lemma:fast} to the external memory model.
\begin{lemma}\label{lemma:extfast}
Let $S$ be a set of two-dimensional points such that $p.y\le \log N$ for all $p\in S$. There exists a linear-space data structure that reports all points 
$p\in [a,b]\times [c,\log N]$ in $O(k/B+1)$ I/Os and supports updates in $O(\log^{\eps}U)$ I/Os. 
\end{lemma}
\begin{proof}
  As in Lemma~\ref{lemma:fast} we divide points into groups $G_i$. Each group, except the last one, contains $\Theta(B\log N)$ points; the last group contains $O(B\log N)$ points. We define $\gmin(i,h)$, $\gmax(i,h)$, and  
$\oR$ in the same way as in Lemma~\ref{lemma:fast}. 
\end{proof}
}

\no{
Using Lemma~\ref{lemma:fast}, we can find all elements $e\in [a,b]$ such that $\hmina(e)\ge height(u_j)$ or $\hmaxa(e)\ge height(u_j)$. We thus obtain the following result.
\begin{theorem}\label{theor:fast}
There exists a linear space data structure that answers one-dimensional color reporting queries in $O(k+1)$ time and updates in $O(\log^{\eps} U)$ time. 
\end{theorem}
We observe that  the structure of Theorem~\ref{theor:fast} cannot be extended to the external memory. The reason is that we have to visit  a number of groups $G_i$ when 
queries to structures of Lemma~\ref{lemma:fast} are processed. We would have to spend $O(1)$ I/Os in each group; therefore obtaining $O(1+k/B)$ query cost with the method 
of Theorem~\ref{theor:fast} is problematic. 
}

\section{Slow Queries, Fast Updates}
\label{sec:slow}
In this section we describe a linear-space data structure that supports color reporting queries in $O(n_{a,b}^{1/2} + \log \log n + k)$ time, where $n_{a,b}=|S\cap [a,b]|$ is the number of elements in the query range $[a,b]$. Although  the query cost is high for large $n_{a,b}$, updates  are supported in $O(\log^{\eps} U)$ time, where $U$ is the size of the universe.

\begin{lemma}
\label{lemma:slowrep}
There exists a linear space data structure that reports all distinct  colors in a query range  $[a,b]$. 
Queries are supported in $O(n_{a,b}^{1/2}+\log\log n +k)$ time where $n_{a,b}=|[a,b]\cap S|$ is the number of elements in the query range $[a,b]$. Updates are supported in $O(\log^{\eps}U)$ time.  
\end{lemma}
\begin{proof}
Let $T$ be a range tree of the set $S$. Leaves of $T$ contain the elements of $S$ in sorted order. The root of $T$ has $\Theta(n^{1/2})$ children. 
Each child of the root node has  $\Theta(n^{1/4})$ children and $\Theta(n^{1/2})$ leaf descendants. 
A node of depth $d$ has $\Theta(n^{g_d})$ children and $\Theta(n^{g_{d-1}})$ leaf descendants where $g_i=(1/2)^{i+1}$. 
Thus the height of $T$ is $O(\log\log n)$. 
As before, $S(u)$ denotes the set of elements stored in leaf descendants of 
$u$.
Recall that $prev(e)$ for $e\in S$  denotes the largest element $e'\le e$  of the same color. The set $C(u)$, $u\in \cT$, contains all elements $e\in S(u)$, such that $prev(e)\not\in S(u)$. 
We maintain a balanced tree $\cT$ and  sets $C(u)$ in all nodes 
$u\in\cT$. Further, all elements of every $C(u)$  are kept in two sorted lists, $P(u)$ and $V(u)$. Elements $e\in P(u)$ are sorted by $prev(e)$; 
elements in $V(u)$ are sorted by their values. Finally, we also maintain 
a data structure, described in~\cite{MPP05}, that supports reporting queries
on $C(u)$ in $O(k+1)$ time and updates in $O(\log^{\eps}U)$ time.  

To answer a query $Q=[a,b]$, we identify the leaves $v_a$ and $v_b$ that 
hold $e_a$ and $e_b$ respectively, where $e_a$ is the smallest element that is greater than $a$ and $e_b$ is the largest element that is smaller than $b$. 
Let $v_q$ be the lowest common ancestor of $v_a$ and $v_b$. Suppose that 
$[a,b]$ covers children $v_{l+1}$, $\ldots$, $v_{r-1}$ of $v_q$ and intersects 
with $S(v_l)$ and $S(v_r)$.  Let $\pi_a$ denote the path from $v_q$ to 
$v_a$. The query answering procedure works as follows. 
\begin{inparaenum}[(i)] 
\item \label{list:slowstep1} We visit all $u\in \pi_a$ and report all elements $e\in C(u)$, 
$e\ge a$, in each $u$  using $V(u)$. 
\item \label{list:slowstep2}
Then we visit all right siblings $u'$ of nodes $u\in \pi_a$ except  $v_l$; 
in every $u'$, we report colors of $e\in C(u')$, $prev(e)< a$, 
using $P(u')$. 
We also report all colors of $e\in C(v_{l+1})\cup\ldots\cup C(v_{r-1}$), 
$prev(e)<a$. 
\item \label{list:slowstep3}
We also report colors of all $e\in C(v_r)$, $e\le b$, using  
$V(v_r)$. 
\item \label{list:slowstep4}
Finally we visit all proper ancestors $w$ of $v_q$; in every 
$w$ we report all elements $e\in C(w)$, $a\le e\le b$, using the reporting 
data structure.  
\end{inparaenum}

Correctness of our procedure can be demonstrated as follows. Suppose that 
$e$ is the leftmost occurrence of some color in $[a,b]$. We consider two different cases. (1) $prev(e)\not\in S(v_q)$. Then $e\in C(w)$ for some ancestor 
$w$ of $v_q$ and it will be reported when $C(w)$ is queried. 
(2) $prev(e)\in S(u_l)$ where $u_l$ is the left sibling of some node $u$ 
on $\pi_a$. Then $e$ is stored in $C(u)$ or in $C(u')$ for a right sibling $u'$ of $u$. Hence $e$ was reported when $u$ (resp.\ $u'$) was visited. 
Each color is reported at most two times. If the  color of an element $e_c$ is reported during step (ii), then  $e_c$ is the leftmost element
of that color in $[a,b]$ because we only output elements $e$ such that 
$prev(e)<a$. If the color of an element $e_c$ is reported during step (i) or step (iv), than $e_c$ is likewise the leftmost element of that color. This is because $e_c\in C(w)$ (resp. $e\in C(u)$) and the leftmost occurrence of a color in $S(w)$ (or $S(u)$) is also the leftmost occurrence in $[a,b]$. 
The only situation when we report the color of an element $e_c$ and $e_c$ is not the leftmost occurrence of that color is during  step~(\ref{list:slowstep3}). We can get rid of duplicates 
by traversing the list of answers and removing elements $e$, such that $prev(e)\ge a$.

The time needed to answer a query can be estimated by counting the number 
of visited nodes.   Let $u_h$ denote the highest node on $\pi_a$ such that 
at least one sibling $u$ of $u_h$ is visited. The number of leaves in the subtree of $u$ is $O(n^{g_d})$, where $d$ is the depth of $u$. We consider all nodes $u_t$ below $u_h$ on $\pi_a$, the total number of siblings of such $u_t$ is bounded by $O(\sum_{i=d+1}^{O(\log\log n)}n^{g_i})=O(n^{g_{d+1}})=O((n^{g_d})^{1/2})$.  
Thus $q(n)=O(n_{a,b}^{1/2}+\log \log n)$, where $q(n)$ is the time needed to answer a query and $n_{a,b}=|S\cap [a,b]|$ is the total number of elements in $[a,b]$.

When an element $e$ is inserted into $S$, we identify the greatest  $e_p\le e$ and the smallest $e_n\ge e$ such that $e_p$ and $e_n$ are  of the same color as $e$. 
We insert $e$ into an appropriate leaf $v_l$ and find the ancestor $u$ of $v_l$, such that $e_p\not\in S(u)$, 
$e_p\in S(parent(u))$. The element $e$ is inserted into the set $C(u)$ and into lists $P(u)$, $V(u)$.  Suppose that $e_n$ was stored 
in a set $C(u_1)$; we remove $e_n$ from $C(u_1)$, $P(u_1)$, $V(u_1)$ and insert it into corresponding secondary structures 
in the node $u_2$. The node $u_2$ is chosen in such  way that 
$e\not\in S(u_2)$ and $e\in S(parent(u_2))$. When the number of elements in 
some node $u$ becomes equal to $2n^{g_{d-1}}$, where $d$ is the depth of $u$, we split $u$ into $u'$ and $u''$. 
When a node is split, we re-build the data structures in $u'$, 
$u''$ and all their descendants. Thus the total amortized cost
incurred by splitting a node is $O((\log\log N)^2)$. 
The total cost of an insertion is dominated by the time necessary to update 
the data structure that supports reporting queries on $C(u)$.
Deletions are symmetric.
\end{proof}

We will also need another result that uses almost the same data structure, 
but reports only the $k$ leftmost colors in the query range.

\begin{lemma}
\label{lemma:slowcount}
There exists a linear-space data structure that reports, for any integer $k$, the $k$ leftmost (rightmost)colors in a query range  $[a,b]$. 
Queries are supported in \no{$O(n_{a,b}^{1/2}\log^2n+\log \log N +k)$} 
$O(n_{a,b}\log N +\log N \log \log N +k\log N)$ time where $n_{a,b}=|[a,b]\cap S|$ is the number of elements in the query range $[a,b]$. Updates are supported in $O(\log^{\eps} U)$ time. 
\end{lemma}
\begin{proof}
The data structure  from Lemma~\ref{lemma:slowrep} can be used to
determine whether the number of points in a query range exceeds a threshold value $\tau$. \no{answer 
the following $tau$-reporting queries: if the number of distinct colors in a query range $[a,b]$ does not exceed a threshold value $\tau$,  we report all distinct colors in $[a,b]$; otherwise we report that the number of colors exceeds $\tau$. To answer a $\tau$-reporting query,}
To compare the number of points, with $\tau$, we  proceed as in Lemma~\ref{lemma:slowrep} and use the fact that every color is reported at most twice. If the number of reported elements in visited nodes exceeds $2\tau$ at some point, we stop processing the query and report that the number of distinct colors exceeds $\tau$. Otherwise, we answer 
a color reporting query and determine whether the number of colors exceeds 
$\tau$. In both cases, a comparison of the number of colors with $\tau$ is performed  in $O(n^{1/2}_{a,b}+\log\log N +\tau)$ time.

We identify the range $[a,b']$ that contains at least $k$ and at most $2k$ colors by binary search.  We start by setting $b'=a+(b-a)/2$ and comparing the number of colors $s(b')$ in  $Q_1=[a,b']$ with $k$. If $s(b')$ is larger than (smaller than) $k$, we move $b'$ to the left (to the right) using a standard binary search procedure.  
After $\log n$ iterations we obtain~$b'$, such that $s(b')=k$. 
Then we answer a color reporting query on $[a,b']$ as in Lemma~\ref{lemma:slowrep}. Each comparison query is answered in 
  $O(n_{a,b}^{1/2}+\log \log N +k)$. $O(\log N)$ iterations take 
$O(n_{a,b}\log N +\log N \log \log N +k\log N)$ time.
\end{proof}

\section{External Memory Solution}
\label{sec:extern}
Our dynamic external memory data structure is based on the same approach 
as the data structure of Theorem~\ref{theor:dyn}. But we need to change 
some of the auxiliary data structures. 
\paragraph{Three-Sided Data Structure for Small Sets}.  
\begin{lemma}\label{lemma:extmem3sid}
There exists a data structure that supports three-sided queries on a set $S$, such that $|S|=O(B\log^6N)$, in $O(1)$ I/Os.
This data structure uses linear space and supports updates in $O(\log^{\eps}N)$ I/Os. 
\end{lemma}
\begin{proof}
If $B>\log^{\eps}N$, the data structure from~\cite{ArgeSV99} gives us the desired query and update bounds because $O(\log_B(|S|))=O(1)$ in this case. 
If $B<\log^{\eps}N$, we implement our data structure as external priority search tree $T$. Every leaf of $T$ contains $B$ points sorted by their $x$-coordinates. As before $S(v)$ denotes the set of all points stored in the leaves of $v$.  Let $Top(v)$ denote the set of $B$ points $p\in S(v)$ with highest $y$-coordinates satisfying  
$p\not \in S(w)$ for any ancestor $w$ of $v$. Each internal node has $\Theta(\log^{\eps} N)$ children. 
We keep a data structure $D(v)$ in each internal node $v$; $D(v)$ contains all points $p\in\cup Top(v_i)$, where the union is over all children $v_i$ of $v$,  and answers three-sided queries. $D(v)$ is implemented as a static data structure described in~\cite{LarsenP12}. When a set $Top(v_i)$ is updated, we re-build $D(v)$.  Since $D(v)$ contains $O(B\log^{\eps}N)=O(\log^{2\eps}N)$ points, it can be re-built in $O(\log^{2\eps}N)$ I/Os.

The query answering procedure is the same as in the external priority tree. For a query $Q=[a,b]\times [c,+\infty]$, let $\pi_a$ and 
$\pi_b$ denote the search paths for $a$ and $b$ respectively. We visit all nodes $u$ on $\pi_a$ and $\pi_b$ and report all points in $Top(u)\cap Q$. Then we visit relevant descendants of nodes $u$ on $\pi_a\cup \pi_b $. In each visited node $u$, we report all points in $Q\cap \cup_{i}Top(u_i)$; we visit a child $u_i$ of $u$ only if $|Q\cap Top(u_i)|\ge B$. Details of the reporting procedure and a proof of its correctness can be found in~\cite{ArgeSV99}. Since the height of our priority tree is $O(1)$, we can update our data 
structure by updating $O(1)$ structures $D(v)$. Replacing $\eps$ by $\eps/2$, we obtain the result of this Lemma.
\end{proof}      

\paragraph{Three-Sided Data Structure on a Narrow Stripe}
Now we show how a data structure of Section~\ref{sec:fast3sid} can be extended to the external memory model. 
\begin{lemma}\label{lemma:extfast}
Let $S$ be a set of two-dimensional points such that $p.y\le \log N$ for all $p\in S$. There exists a linear-space data structure that reports all points 
$p\in [a,b]\times [c,\log N]$ in $O(k/B+1)$ I/Os and supports updates in $O(\log^{\eps}U)$ I/Os. 
\end{lemma}
\begin{proof}
The structure of Lemma~\ref{lemma:fast} cannot be directly extended to the external memory. The reason is that we have to visit 
a number of groups $G_i$ when a query is processed. We would have to spend $O(1)$ I/Os in each group; thus the query time would
 be $O(1+k)$. We need further modifications to obtain the  desired $O(1+k/B)$ query cost.
Our external memory solution is based on increasing the group size. Essentially, we answer a three-sided query in a group $G_i$ only if we know that $G_i$ contains a sufficient number of points from the query range. Otherwise we resort to the data structure that is based on range trees.

As in Lemma~\ref{lemma:fast} we divide points into groups $G_i$. Each group, except the last one, contains $\Theta(B\log^3 N)$ points; the last group contains $O(B\log^3 N)$ points. 
We distinguish between \emph{group-stored} and \emph{directly stored} points. Let $G_i[h]=\{\,p\in G_i\,|\,p.y=h\,\}$. 
A point $p$ is group-stored if $p\in G_i$ and $G_i[p.y]=\Omega(B\log N)$. Otherwise $p$ is directly stored.  
Either all points in $G_i[h]$ for a fixed value of $h$ are group-stored or all points in $G_i[h]$ are directly stored.
We maintain the values of $\gmin(i,h)$ and $\gmax(i,h)$, defined in the same way as in Lemma~\ref{lemma:fast}, 
 with respect to group-stored points.

Let $S_d$ denote the set of directly stored points $p\in S$; let $S_g$ denote the set of group-stored points in $S$. 
All points in $S_d$ are kept in the standard range tree $\bT$ with node degree $\log^{\eps/3}N$. $\bT$ allows us to reduce a 
three-sided query on $S_d$ to $O(1)$ one-dimensional point reporting queries. We refer to~\cite{AlstrupBR00} for details. 
The space usage of $\bT$ is $O(|S_d|\log^{2\eps/3}N)=O(N)$; updates can be implemented by $O(\log^{2\eps/3}N)$ updates of one-dimensional auxiliary structures stored in the nodes of $\bT$. 
We implement these one-dimensional reporting data structures using the result of Mortensen et al.~\cite{MPP05}. Thus $\bT$ supports queries and updates in $O(1)$ and $O(\log^{\eps}N)$ I/Os respectively.
All points in $G_i\cap S_g$ for each group $G_i$ are kept in a data structure $H_i$ that supports three-sided range reporting queries as described in Lemma~\ref{lemma:extmem3sid}.

When a new point $p\in S$ is inserted into $S$, we identify the group $G_i$ where it belongs. If points in $G_i[p.y]$ are group-stored, then 
$p$ should be  group-stored. We compute $h_{r,s}$ for $h=p.y$ and update the values of $\gmin(i,h_{r,s})$, $\gmax(i,h_{r,s})$ as in Lemma~\ref{lemma:fast}. If points in $G_i[p.y]$ are directly stored, then we check how many points are currently in 
$G_i[p.y]$.  If $|G_i[p.y]|=(3B\log N)/2$, then points in $g_[p.y]$ will be group-stored. 
We remove all points of $G_i[p.y]$ from the tree $\bT$ and add them to $H_i$. 
Then we identify the leftmost and the rightmost points in $G_i[p.y]$ and update $\gmin(i,h_{r,s})$, $\gmax(i,h_{r,s})$ for $h=p.y$. 
If $|G_i[p.y]|< (3B\log N)/2$, we continue to store $G_i[p.y]$ directly.  In the latter case we simply add $p$ to $\bT$. 
Deletions are symmetric to insertions. However if points in $G_i[p.y]$ are group-stored and $p$ is deleted, we continue 
to group-store the points in $G_i[p.y]$ if $G_i[p.y]\ge (B\log N)/2$.

For a  query $Q=[a,b]\times [c,+\infty]$, we report all points in $S_d\cap Q$ using $\bT$. We identify groups $G_i$ such that 
$(G_i\cap S_g)\cap Q\not=\emptyset$ in the same way as in Lemma~\ref{lemma:fast}. For every such $G_i$ we report all points 
in $(G_i\cap S_g)\cap Q$ using $H_i$.  If $\min\{\,p.x\,|p\in G_i\,\}\ge a$
and $\max\{\,p.x\,|\,p\in G_i\,\}\le b$ for some group $G_j$, then either  $(G_j\cap S_g)\cap Q=\emptyset$ or $|(G_j\cap S_g)\cap Q|\ge B\log N/2$. Hence, the total cost of answering queries in all $H_i$ is $O(1+|S_g\cap Q|/B)$. 
Thus a query $Q$ is answered in $O(1+k/B)$ I/Os.
\end{proof}

\paragraph{Color Reporting in External Memory.}
We observe that only the case when the block size  $B\le N^{1/16}$ should be considered. If $B>N^{1/16}$, then we can use the reduction of color reporting queries to three-sided queries. The data structure of Arge et al~\cite{ArgeSV99} supports three-sided queries and updates in $O(\log_B N)=O(1)$ I/Os for $B>N^{1/16}$.  In the rest of this section we assume that $B\le N^{1/16}$.
We will only sketch the differences of our data structure and the data structure of Theorem~\ref{theor:dyn}

Our base tree contains $B\log N$ elements in every leaf. We keep the data structure of Lemma~\ref{lemma:extmem3sid} in each leaf node.

\setcounter{section}{0}
\renewcommand\thesection{A.\arabic{section}}

\end{document}